\documentclass[11pt,fleqn]{article}

\usepackage[english]{babel}
\usepackage{amsmath, amssymb, multicol, amsfonts, amsthm, amscd, epsfig, enumerate, xcolor}
\usepackage[latin1]{inputenc}
\usepackage[T1]{fontenc}
\usepackage{graphicx}
\usepackage{geometry}
\usepackage{hyperref}

\newtheorem{theorem}{Theorem}

\newtheorem{proposition}[theorem]{Proposition}

\begin{document}

\author{
    \emph{\textbf{Leonardo Becchetti}}\\
	Department of Economics and Finance, University of Rome Tor Vergata\\
becchetti@economia.uniroma2.it
	\\
	\emph{\textbf{Nazaria Solferino}}\\
	Department of Economics, Statistics and Business, Universitas Mercatorum\\
nazaria.solferino@unimercatorum.it\\
\emph{\textbf{M. Elisabetta Tessitore}}\\
Department of Economics and Finance, University of Rome Tor Vergata\\
tessitore@economia.uniroma2.it }

\title{The Sustainable Future is now: a dynamical model to advance investments in PV and Energy Storage}
\date{}
\maketitle

	\abstract{}
We examine the relationship among photovoltaic (PV) investments, energy production, and environmental impact using a dynamic optimization model. Our findings show that increasing investment in renewables supports both energy generation and ecological sustainability, with the optimal path depending on policy priorities. Our analysis demonstrates that the economic and technological conditions for a transition to PV energy are already in place, challenging the idea that renewables will only become competitive in the future.
We also account for the fact that PV optimality conditions improve over time as storage technology efficiency increases and production costs decrease. In this perspective we find that energy storage may be a more effective policy tool than carbon taxation for cutting emissions, as it faces less political resistance and further strengthens the long-term viability of renewable energy. Policy insights of the paper capture the evolving competitiveness of PV and its role in accelerating the energy transition. They also provide policymakers with strategies to align economic growth with long-term sustainability through renewable energy investments.

	
\vspace{0.5cm}

\section{Introduction}

Climate change, resource depletion, and the environmental consequences of fossil fuel consumption are placing considerable pressure on global economies to decarbonize and adopt sustainable energy practices. Growing public awareness of environmental issues has led to a strong demand for cleaner and more sustainable energy options, marking a significant change in how we choose and use energy (Evans et al., 2009).
Renewable energy, particularly solar power, is central to this transition. In 2020, renewable energy accounted for approximately 29\%  of global electricity generation, with solar power contributing around 11\%  of this share, a figure that continues to increase as technological advancements reduce costs and improve efficiency. The appeal of photovoltaic (PV) technology lies in its scalability and declining costs. According to the International Renewable Energy Agency (IRENA, 2023), the global weighted-average levelized cost of electricity (LCOE) from solar photovoltaics fell by 89\% between 2010 and 2020. Roser (2020) explains that renewable energy costs are declining due to factors like learning curves, where prices drop as installed capacity doubles, making renewables more competitive than fossil fuels, which lack similar trends. Way et al.(2022) using a validated forecasting method, find that a rapid green transition could save trillions compared to a fossil fuel-based system and conclude that, even without considering climate damages, achieving net-zero by 2050 is economically beneficial. 

Renewable energy resources also offer significant social advantages, including improved public health and better employment opportunities. These systems not only create local jobs but also offer consumers more choices while helping to reduce emissions over time (Akella et al., 2009). Economically, renewable energy projects often use local labor, materials, and businesses, supporting regional development through trust funds that reinvest earnings into the local economy (Hicks et al., 2011). In particular, solar energy projects drive job creation across various sectors, from photovoltaic manufacturing to installation, further stimulating regional economic growth. (Akella et al., 2009).

For all these reasons, renewables have emerged as a viable option for governments and companies dedicated to electricity generation, offering substantial environmental, social, and economic benefits. These advantages align with the objectives outlined in the ``2030 Agenda for Sustainable Development" adopted by the United Nations. At the same time, the European Union has established itself as a global leader in renewable energy adoption. In 2021, the European Green Deal set the ambitious goal of making Europe the first climate-neutral continent by 2050, with interim targets of reducing net greenhouse gas emissions by at least 55 percent by 2030. The EU's ``Fit for 55" package, introduced in 2021, includes comprehensive reforms in energy taxation, emissions trading, and renewable energy targets, with the aim of increasing the share of renewables in energy consumption to 40 percent by 2030 (EU Commission, 2021).\\

Despite renewable energy projects such as solar panels play a vital role in reducing pollution and raising awareness about climate change, their implementation can be complex and context-sensitive (Hicks et al., 2011). Some renewable energy projects come with negative externalities that need to be considered. For example, large-scale solar panel installations can require vast areas of land, potentially impacting natural habitats and biodiversity, especially when located in agricultural or ecologically sensitive areas. Despite these challenges, renewable energy initiatives help reduce air a TTnd water pollution through more efficient resource management (Akella et al., 2009). Overall, these projects drive economic development, create jobs, and improve quality of life, while addressing critical environmental challenges.

Considering renewables from an efficiency perspective, hydropower is generally considered the most efficient renewable energy source, followed by wind energy. Photovoltaic and geothermal energy, on the other hand, have historically been considered less efficient due to the limitations of solar cell technology and geographic constraints (Evans et al., 2009). However, recent advancements in battery storage and other energy storage technologies have significantly mitigated these limitations. These innovations enable excess energy produced during the day to be stored and used at night or during cloudy periods, thus improving the reliability of photovoltaic systems and further supporting their role in sustainable energy solutions.\\\\
Despite significant advancements in renewable energy adoption, there remains significant uncertainty regarding the optimal level of investment required in renewable energy technologies. Countries must carefully balance short-term economic needs with long-term sustainability goals, as investments in renewable technologies involve substantial upfront costs but promise significant long-term economic and environmental benefits. The relationship between renewable energy investment, GDP growth, and environmental impact is multifaceted and requires careful planning. Key questions arise related to the optimal amount of investment a nation should devote to renewable energy and the long-term economic and environmental implications of this investment. 
To address these questions, our paper develops a dynamic model that examines the interactions between renewable energy production-specifically solar PV-and environmental impact. By incorporating a key variable representing the investment allocated to PV technology, the model explores energy investment strategies and their temporal trajectories, accounting for the trade-offs between short-term economic returns and long-term environmental sustainability. Although, to our knowledge, previous research has largely examined the relationship between economic growth and environmental quality through frameworks such as the environmental Kuznets Curve or integrated assessment models(for a critique to IAMs see Victoria et al.2021; van de Ven et al.,2025), only a few studies have explicitly modeled the economic viability of renewable energy investments within an intertemporal dynamical framework. This paper aims to bridge this gap, providing also valuable insights for policy makers who wish to accelerate the adoption of renewable energy while maintaining sustainable economic growth.  To this aim, in addition, the paper highlights the critical role of energy storage technologies in shaping the effectiveness of PV investments. Taking into account also traditional storage solutions, the model underscores how coupling PV investments with advanced storage technologies can mitigate intermittency issues, improve energy efficiency, and further align economic growth with environmental objectives. These insights demonstrate the importance of integrating investment strategies for both energy generation and storage to ensure a sustainable and resilient energy transition.
\\The remainder of the paper is structured as follows. Section 2 provides a literature review, establishing the context and motivation for this study and how it bridges part of the existing gap in the state-of-art of the literature. Section 3 introduces the model and details its key components, including law of motions energy production and environmental degradation, and presents simulations of investment strategies under various scenarios. Section 4 extends the analysis by incorporating storage-efficient innovations. Finally, Section 5 offers a discussion and conclusions, highlighting the broader implications of our findings for energy policy, with insights aiming to assist policymakers in aligning their energy strategies with both economic growth and environmental sustainability objectives.
\vspace{0.5cm}

\section{Literature Review}

The interconnections among renewable energy investments, economic growth, and environmental sustainability have attracted significant attention in both academic and policy domains over the past two decades. Renewable energy, especially solar and wind power, is increasingly recognized as essential for addressing climate change and reducing reliance on fossil fuels. As the global energy landscape transitions to decarbonization, a growing body of the literature has examined the economic and environmental implications of renewable energy adoption, emphasizing policy design, cost-benefit analysis, and technological innovations.
Much of the existing literature focused on the macroeconomic impacts of renewable energy investments. Several studies highlight the positive correlation between renewable energy deployment and GDP growth, particularly in countries with supportive policy frameworks. For instance, Marques and Fuinhas (2011) demonstrated that renewable energy consumption significantly contributes to economic growth, particularly in nations less dependent on fossil fuels. Apergis and Payne (2010) similarly identified a positive relationship between renewable energy consumption and economic growth in OECD countries, reinforcing the idea that green investments can drive economic expansion. Stern and Vos (2006) underscored the economic benefits of transitioning to renewable energy, noting that the costs of inaction due to climate damage far outweigh the investments needed for adopting sustainable energy sources. Sadorsky (2009) further explored the relationship between renewable energy consumption and economic growth in emerging markets, concluding that green energy supports economic development while reducing emissions. Pablo-Romero and Sanchez-Braza (2017) conducted a comprehensive review of the renewable energy-growth nexus, highlighting the critical role of government policies in promoting technological innovation and investment in the energy sector. Sanchez et al. (2021) examined the feasibility of renewable energy in off-grid communities, using a mathematical optimization model to balance seasonal availability, economic costs, and environmental impact.

Beyond economic outcomes, a substantial portion of the literature addresses the environmental benefits of renewable energy. Numerous studies show that increased investments in renewable sources, particularly solar and wind, significantly reduce greenhouse gas (GHG) emissions and air pollution. Del Rio and Burguillo (2008) emphasized that renewable energy supports international climate agreements, such as the Paris Agreement, while fostering local economic development and employment. Gayen et al. (2023) reviewed the environmental impacts of renewable energy, highlighting how innovations such as solar windows and smart grids could revolutionize the energy sector. Ferhi and Helali (2024) assessed renewable energy's role in mitigating CO2 emissions in OECD countries, finding that it enhances both economic growth and human development by reducing per capita emissions and promoting healthier, more prosperous societies.

Research on policy interventions further underscores the critical role of government action in advancing renewable energy adoption (Skovgaard and Van Asselt, 2019). Policies such as feed-in tariffs (FITs), renewable energy credits, and carbon pricing have proven effective in promoting green technology. Lund (2007) and Menanteau et al. (2003) showed that well-designed frameworks reduce financial barriers and correct market failures, such as the underpricing of environmental externalities. Rastegar et al. (2024) provided a classification of environmental policies' effects on renewable energy innovation, while Wang et al. (2019) emphasized governments' pivotal role in mobilizing investments and implementing effective policies. The debate around policy effectiveness, however, persists, as some studies argue these measures enhance green innovation (Bersalli et al., 2020; Kilinc-Ata, 2016), while others, such as Nesta et al. (2014), question whether the high initial costs of renewable energy technologies hinder their adoption.

Technical studies have also advanced understanding of renewable energy integration into low-carbon systems. Ostergaard et al. (2022) provided a systematic review of renewable energy's role in decarbonization, noting improvements in cost reductions, emissions targets, and system efficiencies, all contributing to sustainable energy solutions. Such research emphasizes that large-scale deployment of renewable energy is increasingly feasible.

Despite the extensive body of research, there remains limited focus on economics dynamic models that incorporate investment decisions within an intertemporal framework, as most studies primarily address technical issues.  For example, Li et al. (2021) analyzed energy substitution in China using a continuous dynamical system, examining the interplay between increasing renewable energy shares and reducing coal consumption. Semmler et al. (2024) developed a dynamic portfolio model to investigate how short-term investment horizons influence green investments, focusing on the implications for portfolio behavior. Islam et al. (2024) employed a system of ordinary differential equations (ODEs) to explore approaches to reducing GHG emissions, emphasizing the mechanisms for achieving environmental sustainability. Similarly, Denhg et al. (2024) discussed the challenges associated with intertemporal resource allocation models, underscoring their practical constraints due to resource scarcity but highlighting their importance for sustainable development policy.

While static models provide valuable insights, they may fall short in capturing the long-term economic trade-offs and feedback effects essential for understanding how renewable energy policies affect economic growth and environmental quality over time. Our paper seeks to address this gap by proposing a dynamical model that incorporates renewable energy investment as a central variable, offering a deeper perspective on its economic and environmental implications. In particular, we aim to assess whether and under which conditions i) the conditions for a successful transition to solar energy (PV) are already in place and actively underway and ii) the transition can  be accelerated through targeted investments in storage technologies, which offer a more effective solution than fiscal policies such as carbon taxes,as the latter face significant resistance from both businesses and consumers, who are concerned about the impact of such measures on final prices, making storage investment a preferable approach.  This approach aligns with the argument of Becchetti et al.(2022), who emphasize that renewable energy is not only the most cost-effective option but also a key driver of sustainability, economic stability, and energy independence. Inspired by these insights, this paper develops a framework that explicitly accounts for the role of renewable energy in shaping long-term economic and environmental dynamics.

\vspace{0.5cm}

\section{The model}
In our dynamical model we aim to analyse and find the optimal level of investments in renewable energy $I$  taking into account their economic costs and benefits, energy production and the environmental impact. For simplicity and without lack of generality we restrict our focus on photovoltaic energy which represented in 2023 the 73 percent of all new renewable plants at world level, with renewable plant accounting for 87 percent of all installed capacity. To do so we need to model the dynamics of the total energy production $E_t$, derived from photovoltaic energy netted from the negative environmental impact or damage $D$.

We therefore assume that renewable energy production over time is determined by investments $I_t$ and grows according to the following law of motion:
\begin{equation}
	\dot{E_t} = \eta I_t - \delta_E E_t
\end{equation}

where $\eta$ is a productivity index and measures how much energy is produced per unit of investment or used resources and $\delta_E$ captures the effect of depreciation of the existing energy plants over time.

The environmental impact (carbon emissions) $D$ depends on the amount of energy produced and the sources used to produce it, specifically:
\begin{equation}
	\dot{D_t}=- \omega  \eta I_t+ \delta_D D_t
\end{equation}

where $\delta_D$ measures the rate of growth of carbon emissions according to the existing energy plants and $ \omega=(\omega_f -\omega_{pv})$ measures the environmental impact savings resulting from the use of PV compared to fossil fuels (oil or natural gas).
In other words, we are considering, based on the standard Life Cycle Assessment which considers all stages of the production process (from extraction of raw materials to recycling of waste), that renewables are not completely free from environmental impacts, still producing from 50 to 100 less CO2 emissions than fossil fuels (see data in the footnote). In fact, although they have marginally zero emissions during plant activity, it is important to consider some climalterating effects related, for example, to the manufacturing and installation of photovoltaic systems, land use, the life cycle of technologies, or the uncertainty associated with photovoltaics, due to solar intermittency or energy storage cost.

We assume a positive sign for $\omega$, that is $\omega_{pv}<\omega_f$, that reasonable means that there is a small impact negative, as photovoltaics do not produce pollution as much as fossil, moreover from now on, we assume that parameters satisfy $ \rho > \delta_D $ and $c, \eta, \beta,  \gamma > 0$.

We conveniently assume that the global community want to maximize an objective function combining return from energy investment and energy efficiency, measured as the difference between energy production and environmental impact on global warming. The goal of our dynamic model is therefore to find the optimal value of $I $ by solving the following maximization control problem: 

\begin{equation}
	\label{opt*}
 \max\int{[(r I_t- c I_t^2)+ (\beta E_t-\gamma D_t)]} e^{-\rho t}dt
\end{equation}
subject to
\begin{equation}
		\label{sist*}
\left\{
\begin{array}{ll}
\dot{E_t} = \eta I_t - \delta_E E_t&E_0>0
\\\\
	\dot{D_t}=- \omega  \eta I_t+ \delta_D D_t&D_0>0
\end{array}
\right.
\end{equation}
where $r=r_{pv}/r_f$ , $ c=c_{pv}/c_f$ and $\eta=\eta_{pv}/\eta_f$ represent economic return, cost and productivity of PV calculated with respect to the fossil fuels
s.t.
 $E_t \geq 0, I_t \geq 0$.
\begin{proposition}\label{prop:I*}
The optimal control $I^*$ of the optimal control problem $(\ref{opt*})-(\ref{sist*})$ is
\begin{equation}
	I^*=\frac{r}{2c} +\frac{\eta }{2c}\left [\frac{\beta }{\delta_E +\rho }+\omega \left(\frac{\gamma }{\rho -\delta_D }\right )\right ].
\end{equation}

\end{proposition}
\begin{proof}
	Considering the Hamiltonian associated to the intertemporal optimization problem $(\ref{opt*})-(\ref{sist*})$, we find the following optimal value for $I^*$, (see the Appendix).
	
\end{proof}

Therefore, under the parameter restriction $\delta_D < \rho$,  which ensures finite depreciation costs, the proposition asserts that the optimal investment grows in the weight given to energy production and reduction of environmental impact, in the economic revenues from photovoltaic investment and its relative higher ecological efficiency (reduced emissions with respect to fossil fuels). The optimal investment grows is the rate of growth of emission is higher and falls if the rate of depreciation of existing energy plants is higher.
Specifically, the above condition suggests that, regardless of the direct economic return, the investment in photovoltaics remains justifiable. This is because the energy produced and the savings from environmental damage are particularly significant when the rate of environmental regeneration is slow, as the long-term benefits of reducing environmental damage and generating clean energy become increasingly important over time, even if they are undervalued in the short run due to the high discount rate $\rho$.

Despite this result may seem intuitive, the purpose of this section is to demonstrate that, given baseline parameters and real-world data, the above conditions hold, and projections indicate further improvements in the future. As a result, we may have reasonably reached a point where there are no valid reasons to continue relying on fossil fuels instead of prioritizing the transition to PV energy. Our approach aims to move beyond a general claim about the advantages of photovoltaics to a rigorous demonstration using an economic model and real-world data, showing that necessary and sufficient conditions for the transitions to PV are already met today. This challenges the common perception that renewables will only become fully competitive in the future and suggests that there are no longer economic or environmental justifications for delaying the transition. In doing so, our work also raises a critical question: if the transition is already viable, why has it not yet occurred? This suggests that probably non-economic barriers may be the main obstacles to a full shift towards PV energy.

To perform our simulations, we rely on current and projected data on energy costs, efficiency, CO2 emissions, and global energy capacity from reputable sources (Table 1).

\begin{table}[ht]
\caption{Initial Parameters values}
\centering
\begin{tabular}{|p{2cm}|p{6cm}|p{3cm}|}
\hline
\textbf{Parameters} & \textbf{Value} & \textbf{Source} \\
\hline
c &  0.44 (30-60 euro/MWh PV vs 70-200 euro/MWh fossil fuel) & IRENA (2023) \\
\hline
$\delta_E$ & 0.05 (yearly average value \%) & IRENA (2023) \\
\hline
$\delta_D$  & 0.04 (yearly average value \%) & Ember Climate (2023) \\
\hline
$\gamma$ &  0.19-0.20 (kgCO2/kWh ) & Ember Climate (2023) \\
\hline
r  &  3.4 (PV is always more competitive than fossil fuel) & IRENA (2023) \\
\hline
$\eta$  &  1.4  & IEA (2023); Aramendia et al.(2024) \\
\hline
\end{tabular}
\label{tab:parameters_simulations}
\end{table}

\noindent  We calculate the equilibrium values of $I$, from the Equation 4 for various combinations of parameters and we also show how the path towards the stationary-state of $E$ and $D$ change accordingly. In our simulations initially, the relative performance $r$ and cost $c$ of photovoltaics compared to fossil fuels are calculated as follows: the relative performance $r$, representing the overall economic benefits of photovoltaics, including energy production, job creation, and economic growth, compared to fossil fuels,  was estimated at r = 3.04, based on the broader socio-economic advantages generated by the photovoltaic sector. Studies by IRENA (2023) highlight that renewable energy projects create approximately 2.5 times more jobs per MWh than fossil fuels, demonstrating the additional benefits photovoltaics bring to the economy. This substantially confirm that the necessary condition in Proposition 1 actually holds. Moreover, in our model we start by assuming  $\eta=1.4$ , that represents the relative energy efficiency of photovoltaic (PV) systems compared to thermoelectric power plants. This is also a realistic value, as given that electricity production from fossil fuels typically has an efficiency of 35-50\%, every MWh produced by photovoltaic systems replaces approximately 1.4 MWh of primary fossil energy. This value is consistent with estimates from the IEA (2023) regarding the energy substitution potential of renewable sources. Recent studies have shown that renewable energy systems, including PV, offer higher EROI compared to fossil fuels. For instance, a study published in Nature Communications found that fossil fuels have an EROI of approximately $3.5:1$, whereas renewable energy systems, such as wind and solar photovoltaics, have higher EROI values, indicating more net energy returned to society per unit of energy invested. (Aramendia et al.2024)
The relative cost $c$ is the ratio of photovoltaic energy costs to fossil fuel energy costs. With the current cost of solar energy at approximately  euros 70/MWh and fossil energy at  euros 160/MWh, the relative cost was calculated as $c = 70 / 160 = 0.44$.
Projections for the future, however, indicate significant changes to both $r$ and $c$. For $r$, the relative performance of photovoltaics is expected to increase due to the ongoing reduction in production costs and the expansion of the renewable energy market. According to the International Energy Agency (IEA), photovoltaic costs are projected to decrease by 10-15\% every five years through 2030, primarily due to technological advancements and economies of scale. Additionally, the renewable energy sector is expected to generate 27 million jobs by 2030, up from the current 12 million, as per the IRENA 2023 Renewable Energy and Jobs Report. These factors collectively suggest that the relative performance of photovoltaics could increase by approximately 25\%, bringing its value from 3.04 to 3.8.
For $c$, the introduction of a carbon tax on fossil fuels is expected to make photovoltaics even more competitive. A carbon tax of  euros 50-100 per ton of CO2 could increase the operational costs of fossil fuel energy by 10-15\%, as noted by studies from the European Commission (2024) and ENEA (2016). This would reduce the relative cost of photovoltaics compared to fossil fuels. Assuming a 12\% reduction in the relative cost due to these changes, $c$ would decrease from 0.44 to 0.39. In our simulations, we accounted for both the initial costs of solar technologies and the impact of carbon tax policies. The results are realistic because the calculated values for the optimal investment $I$ (11.6364 euro/MWh without the carbon tax and 14.1026 euro/MWh with the carbon tax) fall within the typical cost range for large-scale photovoltaic projects. The cost of installing photovoltaic systems typically ranges from 800 to 1200 euro/kW, with energy production around 1000-1200 kWh per kW per year, which translates to roughly 10-15 euro/MWh over the lifespan of the system. Therefore, our results are consistent with these benchmarks, indicating a reasonable cost for the energy generated. Additionally, the increase to 14.1026 euro/MWh when considering carbon tax is plausible, as carbon pricing policies make fossil fuel-based energy more expensive, thus increasing the competitiveness of renewable energy sources like photovoltaics. This reflects the general economic trend of renewable energy becoming more cost-competitive as carbon taxes drive up fossil fuel costs. Overall, these values are well within the range of current real-world costs for photovoltaic energy.

\vspace{0.5cm}

Figures 1 and 2 illustrate how the curves of $E$ (energy production) and $D$ (environmental impact), respectively, evolve as the optimal investment ($I$) changes. The shifts reflect the adjustments in the performance and cost dynamics of photovoltaics, incorporating both future projections and the impact of carbon tax on fossil fuels.\\
In Figure 1 the growth of energy production over time is represented by two concave curves: a green curve and a blue curve. The green curve shows the growth in energy production when a carbon tax is implemented, and the efficiency of photovoltaic systems improves, while the blue curve represents the baseline scenario without these measures. Both curves are concave, reflecting diminishing returns in energy production growth over time, as they flatten out. The green curve stabilizes at values above 200, while the blue curve levels off around 150. This flattening occurs around time $t\rightarrow 20$ years in both curves. This indicates that the combined impact of the carbon tax and improved photovoltaic efficiency leads to approximately a 33\% increase in total energy production compared to the baseline scenario.
This scenario is quite realistic, when assuming all other factors remaining constant. Improvements in photovoltaic efficiency, driven by technological advances, could increase energy production by 10-15\% in the next 20 years. A well-designed carbon tax can further incentivize the shift to renewable energy, making the observed 33\% increase plausible within this time frame. This time frame of 20 years seems quite plausible for such changes to take effect, and it aligns with typical planning horizons for both technological advancements and policy shifts. These results underscore the potential of combining technological innovation with economic policy to accelerate the transition to sustainable energy systems.
Similarly, the curve representing environmental impact in Figure 2 shows negative values, which correspond to a reduction in environmental damage such as lower emissions. Over time, the environmental impact decreases and stabilizes, reflecting the positive effects of the carbon tax and improved photovoltaic efficiency.\\\\
Notice that the cost of fossil-based electricity varies by technology and service type. While PV does not replace a single category of fossil generation, it offsets a mix of baseload and peaking power. The choice of $c=0.44$ reflects an average of these two, recognizing that peaking plants, typically gas-fired, have higher electricity costs than baseload coal or gas plants. Moreover, the cost ratio between PV and fossil generation is not fixed; as PV penetration increases, its marginal value may decline due to grid integration challenges and may change among to geographical areas.
Despite these factors, $c=0.44$ aligns with literature estimates for scenarios where PV competes with a typical fossil mix without reaching saturation levels that would significantly alter its competitiveness. Given the many variables involved, this value should be seen as an indicative reference rather than a fixed parameter, requiring further refinements for specific contexts.
Nevertheless to take into account to several costs and productivity  variations, we examine how fluctuations in such key parameters affect the optimal path of energy production and environmental sustainability. Figure 3 shows the optimal investment path in photovoltaic technology under varying conditions of cost and productivity improvements. The graph takes the shape of a folded sheet, with cost changes represented along the X-axis, photovoltaic productivity (efficiency) along the Y-axis, and the optimal investment levels on the Z-axis.

\vspace{1cm}

In Fig.3 the yellow region, located at the more vertical part of the graph, corresponds to higher optimal investment values and higher efficiency, highlighting scenarios where both technological costs are lower, and productivity advancements are significant. In contrast, the dark blue region near the horizontal section of the graph reflects lower investment levels, associated with high costs and low efficiency, indicating suboptimal conditions for investment in photovoltaics. According to projections, the relative efficiency of photovoltaics is expected to increase. This allows us to assert that the world economy currently falls within this area of the figure and that investments in PV are expected to grow in the future.
This surface plot illustrates how the interaction between technology costs and productivity improvements influences the required investment to achieve optimal energy production and sustainability. It highlights the importance of both cost reduction and productivity enhancement in making photovoltaic investments more attractive and environmentally effective.

\vspace{1cm}

\section{The impact of Energy Storage Technologies}

To introduce Energy Storage Technologies and their effects into the model, we need to consider their impact on both the dynamics of energy production and environmental issues. In general, energy storage technologies allow for the accumulation of energy produced by renewable sources during periods of low demand or high production (e.g., during the day) and its use when demand is high or renewable energy production is low (such as during the night).
We assume that investment in solar panel in this case is reduced by an amount $s$ spent for the storage technologies, as the total budget available must be distributed along the two expenses, so the total investment in solar panels now becomes $I_t(1-s)$

In this version of the model including storage technologies the state-equations of $E_t$ and $D_t$ are:

\begin{equation}
\dot{E_t} = \eta I_t(1-s) - \delta_E E_t
\end{equation}
\begin{equation}
\dot{D_t}= -\omega  \eta I_t(1-s)-q S_t + \delta_D D_t
\end{equation}
where $q$ measures the negative environmental impact associated with storage.

The stored energy over time $S_t$ evolves based on the efficiency of the storage system $\eta_S$, that is, the energy released by the storage system at time $t$, on the investment $sI_t$. The variable $\delta_S$ is the the degradation rate over time or energy loss in the storage system. The dynamic equation for $S_t$ can be modeled as follows:

\begin{equation}
\dot{S_t}= \eta_S (s I_t) -\delta_S S_t
\end{equation}

\vspace{0.5cm}
Considering the storage costs $c_S$, the optimization problem now becomes:

\begin{equation}
	\label{opt1}
 \max\int{[(r I_t(1-s)- cI _t^2 (1-s)^{2}+ (\beta E_t-\gamma D_t+\varsigma S_t-c_S s I_t)]} e^{-\rho t}dt
\end{equation}

subject to
\begin{equation}
	\label{sist1}
	\left\{
	\begin{array}{ll}
	\dot{E_t} = \eta I_t(1-s) - \delta_E E_t&E_0>0
		\\\\
	\dot{D_t}= -\omega  \eta I_t(1-s)-q S_t + \delta_D D_t&D_0>0
	\\\\
	\dot{S_t}= \eta_S (s I_t) -\delta_S S_t&S_0>0
	\end{array}
	\right.
\end{equation}

Including the effect of energy storage in the model implies that the availability of stored energy can increase the resilience of the energy system and improve demand management, potentially reducing the need for fossil energy sources. This can lead to more sustainable growth by increasing energy efficiency.

\begin{proposition}
	
	Assume $0<s < \frac{r}{r+c_s}$,  then the storage system is economically viable and the optimal value of $I$ with storage technologies associated to the optimal control problem $(\ref{opt1})-(\ref{sist1})$ is:
	
	\begin{equation}
		$$
		I_1^*=\frac{r(1-s)-c_s s}{2c(1-s)^2}+\frac{1 }{2c(1-s)} \left [\frac{\beta \eta }{\delta_E +\rho }+\frac{\omega  \gamma \eta }{\rho -\delta_D  }+\frac{(\varsigma+\frac{\gamma q}{\rho-\delta_D })\eta_S}{\delta_S +\rho } \right ]
		$$
	\end{equation}
\end{proposition}

\begin{proof}
Considering the Hamiltonian associated to the intertemporal optimization problem $(\ref{opt1})-(\ref{sist1})$, we find the following optimal value for $I^*$, (see the Appendix).

\end{proof}

We can compare the optimal values obtained in the two models

\begin{proposition}
	
	Assume:   
	$$
	0<s < \min \left\{\frac{r}{r+c_s};\frac{(\varsigma+\frac{\gamma q}{\rho-\delta_D })\eta_S}
	{(\varsigma+\frac{\gamma q}{\rho-\delta_D })\eta_S+c_s(\delta_S+\rho  )}\right\},
	\mbox{then}\quad I_1^*>I^*.
	$$
\end{proposition}
\begin{proof}
	
$$
\begin{array}{l}
I_1^*-I^*=
\frac{s}{1-s}\left[\frac{r}{2c}+\frac{\eta }{2c}
\left(
\frac{\beta }{\delta_E +\rho }+\frac{\omega \gamma }{\rho -\delta_D }
\right)\right]
+
\\\\
+\frac{1}{2c(1-s)}
\left[
\frac{(\varsigma+\frac{\gamma q}{\rho-\delta_D })\eta_S(1-s)-c_S(\delta_S+\rho)s}{(1-s)(\delta_S+\rho)}
\right].
\end{array}
$$
\end{proof}

  Proposition 2 indicates that for storage to be a worthwhile addition, the economic return must justify the extra cost, meaning that storage becomes more attractive when either the cost of storage decreases or the return on energy investments increases. Integrating storage is only advantageous compared to the model without it only if the fraction of resources allocated to storage does not exceed a certain threshold, which depends on the economic return and the cost of storage. If it is too high, the additional cost of storage outweighs the benefits, making the investment without storage more attractive. Essentially, storage is justified only if its share remains sufficiently low compared to the photovoltaic returns maintaining the overall efficiency of the investment. Moreover, as the productivity of storage improves, for example, through higher efficiency in capturing and storing energy, the threshold for the fraction of resources allocated to storage rises. In other words, when the productivity of storage and photovoltaic energy improves, or when the costs of storage decrease, the threshold increases, allowing for a higher allocation to storage while still maintaining a positive return on the overall investment.\\
In Figures 4 and 5 we compare the results in terms of energy production with and without storage technologies. Results integrating energy storage solutions show that they improve the effective efficiency, allowing for better energy management, especially during periods of low solar generation. The resulting impact of enhanced storage capacity is a more stable energy supply, which directly influences the investment dynamics for PV systems. As storage solutions become more prevalent, the investment share for PV could increase significantly, reinforcing the case for renewable energy as a viable alternative to fossil fuels. The parameters related to energy storage used in the model are as follows. The investment in storage, represented by s, is set to 0.3, meaning that 30\% of the total investment is allocated to storage. The efficiency of the storage system, $\eta_S= 0.85$, indicating that the system operates at 85\% efficiency, which is typical for modern lithium-ion batteries. The depreciation rate of storage, is set to $\delta_S= 0.02$, meaning that the storage capacity decreases by 2\% per year. The environmental impact of storage is determined by factors like material extraction, energy consumption, and emissions throughout the lifecycle of the storage system. The chosen value $q=0.04$ reflects typical values for such systems, although it may vary depending on technology and regional conditions (to assess how storage technology influences overall environmental performance, the value will be varied alongside the parameter $\varsigma$ in Figure 6, analyzing different combinations of these factors). Finally, the weight of storage in the utility function, $\varsigma$, is set to $0.6$, implying that the storage has a relatively low importance in the optimization compared to other factors. These values are realistic and aligned with current storage technologies and their characteristics.\\
Figure 4 displays three curves with time $t$ on the horizontal axis and energy production (MWh) on the vertical axis. The time values range from 0 to 20 years, and the energy production values range from 0 to 200 MWh. All curves are concave, but the first curve flattens out as it approaches the point (20, 200). The upper curve in magenta represents energy production using storage technologies, which allows for achieving higher energy outputs in less time, such as reaching more than 200 MWh in under 20 years (at $t=12$). This highlights the advantages of energy storage in accelerating energy production.
Based on recent technological advancements and current storage capabilities, it is realistic that a storage system could enable the achievement of 200 MWh in less than 20 years and, more specifically, in 12 years compared to the 18 with carbon tax considering the growing adoption and expansion of these technologies.
This last result is of interest in two key ways. First, it reveals that storage may be an alternative to the carbon tax as a primary policy tool for reducing emissions. Even if climate policies often prioritize carbon pricing, our findings suggest that investing in storage could be even more effective in enhancing sustainability. This finding has a significant economic and political implications. The carbon tax alone may be less effective because fossil fuel producers can often pass the additional costs onto consumers, reducing its impact on overall emissions. This can lead to higher energy prices without necessarily accelerating the transition to renewables. In contrast, storage directly addresses the technical challenge of intermittency, which limits the efficiency of solar and wind power. By enabling surplus energy to be stored and used when needed, storage reduces reliance on fossil fuel-based backup power, making renewable energy more stable and self-sufficient. If storage proves more effective than carbon taxation in reducing environmental damage, it could be more politically viable, as carbon taxes often face resistance. This could justify stronger public support for incentives and investments in storage technologies rather than relying solely on fiscal measures.\\
Notice that, the graph assumes costs of energy storage equal to those of photovoltaic panels, but this value can serve as a good starting point for simulations exploring scenarios where storage is not yet competitive compared to fossil fuels. However, if the goal is to represent a relatively mature storage technology with reduced costs, it would be more realistic to explore the potential of energy storage in a future where technological advancements or cost reductions make storage more affordable. Therefore, for a more comprehensive comparison, we show in Table 2 how various productivity levels and different values of $c_s$ to assess the impact of these parameters on energy production, offering a more detailed view of how energy storage could perform under different conditions.\\
The first column shows the energy production under the initial settings, where no modifications or updates are applied.  The second column reflects the energy production after the application of a carbon tax,and the third column shows the scenario with storage technologies.
The table covers a time span from 0 to 20 years. From the results, we can observe how energy production evolves when adjusting the parameters and highlight the importance of selecting the optimal values for$\eta_S$ and $c_s$ to maximize energy output. When costs ( storage efficiency) are sufficiently low, specifically when the cost parameter $c_s \leq 0.5$ and $\eta_S>0.2$, the energy storage scenario consistently yields better results than the carbon tax scenario, which is represented in the second column. The analysis shows that energy storage outperforms the carbon tax, highlighting its future potential as a competitive solution for energy production.

Our findings outline the synergistic effect between photovoltaics and storage. Without storage, solar power is limited by intermittency, reducing efficiency and leading to energy waste.\\

In addition, Figure 6 shows the environmental impact of energy storage compared to PV-only and Carbon Tax scenarios for different combinations of $q$ and $\varsigma$. The scatter plot reveals that in most cases (green points), storage leads to a lower environmental impact, suggesting that under a wide range of conditions, energy storage may be a more effective strategy. However, some scenarios (red points) indicate a higher impact, highlighting the importance of parameter selection in maximizing the benefits of storage solutions

Our results indicate that storage not only improves renewable energy production but also mitigates environmental impact by reducing the reliance on fossil fuels to balance fluctuations. Finally, the ability to achieve higher energy production in a shorter time with the support of storage technologies is a well-documented feature, as it allows for reducing inefficiencies and optimizing energy output (IEA, 2023a and 2023b).

By reflecting on existing storage technologies, it appears that the conditions for making storage economically viable are realistic and generally satisfied. For example, Lithium-ion batteries, with high productivity (90-95\%) and relatively low projected costs (euros 70-100/MWh), provide a significant incremental productivity that easily outweighs additional costs. Furthermore, their low degradation rate (2-4\% annually) ensures long-term stability, reinforcing their economic appeal. Flow batteries and thermal storage systems, despite their lower productivity (65-85\%) compensate with a lower degradation rate (1-2\% annually) and greater longevity. This makes them particularly suitable for stationary, long-term storage applications in industrial and large-scale utility contexts where durability is a very crucial factor. In other words, all the technologies considered exhibit characteristics that make storage economically and technologically advantageous, provided that the overall system productivity outweighs relative costs and the contribution of storage compensates for intertemporal losses.

\section{Conclusions}
Our study offers a thorough examination of the dynamics of investments in photovoltaic (PV) energy as an alternative energy sources of fossil fuels. By employing a dynamical system, we aim to contribute to the existing literature on energy economics and provide valuable insights for policymakers and stakeholders involved in the transition to renewable energy. Our findings highlight the economic and environmental advantages of prioritizing PV investments, indicating that an optimal investment significantly boosts energy production and reducing carbon emissions at the same time.

A key contribution of this research is its approach to modeling the dynamics of renewable energy sources. By integrating various factors, including environmental sustainability and energy efficiency, this study aims to provide a comprehensive framework that reflects the complexities of energy markets. Such an integrated perspective is essential for understanding the intricate relationship between economic imperatives and environmental goals.

The practical implications of the findings demonstrate that robust investment in PV yields relevant positive results advocating for policies that incentivize renewable energy adoption. Discussions around economic policies, such as carbon taxes and expenses in improving storage mechanisms, underscore the significant influence that regulatory frameworks can have on investment decisions, suggesting that strategic government intervention is crucial in promoting the transition to cleaner energy sources.

We show that, based on real-world data and baseline parameters, the necessary conditions for the transition to photovoltaic (PV) energy are already in place, with future projections indicating even further improvements. This means there is no longer a strong justification for continuing to rely on fossil fuels instead of prioritizing PV energy.

Rather than simply arguing that photovoltaics have advantages, our approach provides a rigorous economic model supported by real data, demonstrating that both the necessary and sufficient conditions for the transition are already in place. This challenges the common belief that renewables will only become fully competitive in the future and suggests that economic and environmental factors no longer justify postponing the shift.

This implies that non-economic barriers are likely the main obstacles preventing a full transition to PV energy.

In summary, the results of this study support the prioritization of investments in photovoltaic (PV). The optimization model and simulations reveal significant economic and environmental benefits associated with increased investments in PV technologies, indicating a potential shift in energy investment paradigms.

The  advantages of enhanced energy production annually and substantial reductions in carbon emissions make PV energy a leading candidate for sustainable economic development. These findings emphasize the importance of supportive economic policies, technological advancements, and energy storage solutions in facilitating the transition to renewable energy.

 Our study suggests in this perspective that energy storage could be a viable alternative to carbon taxes as a strategy for lowering emissions. While carbon pricing is commonly emphasized in climate policies, our results indicate that directing investments toward storage might be an even more effective way to achieve sustainability.

This insight carries significant economic and political weight. Since carbon taxes often face resistance, storage could offer a more practical solution if it proves to be more effective in mitigating environmental damage. This strengthens the argument for increased public support, including incentives and investments in storage technologies, rather than relying primarily on taxation.

These findings highlight the need for a comprehensive approach that combines economic policies, technological innovation, and energy storage to accelerate the shift toward renewable energy.

In particular, the incorporation of energy storage systems is crucial for enhancing the reliability and efficiency of PV technologies. Advanced storage solutions, such as solid-state batteries, hydrogen-based systems, and thermal storage, can address the intermittent nature of solar energy, ensuring a stable and continuous energy supply while optimizing resource utilization. As the energy landscape evolves, it is essential for policymakers and stakeholders to cultivate an environment that not only encourages investment in PV systems but also prioritizes the development and integration of innovative storage technologies. This dual focus will be instrumental in creating a resilient energy infrastructure capable of satisfying both economic needs and environmental objectives.

Looking ahead, this study opens several directions for future research. Subsequent investigations could examine the integration of additional renewable technologies- such as wind or biomass- and to compare these and PV with other energy sources typically considered less sustainable like nuclear power (Moccia,2025). This would enable a more comprehensive evaluation of their respective economic and environmental impacts. Furthermore, examining the socio-economic implications of transitioning to renewable energy-such as job creation and community impacts-could enrich the discussion on energy transitions, ensuring that the shift to sustainable energy is both economically viable and socially equitable.

\vspace{2cm}
\textbf{Acknowlegements}. The authors thank L.Moccia from ICAR-CNR for his useful suggestions. Nonetheless, any errors remain the sole responsibility of the authors.

\vspace{1cm}
\textbf{Declaration of Interest}. The authors declare that they have no known competing financial interests or personal relationships that could have appeared to influence the work reported in this paper.

\newpage
\vspace{1cm}
\textbf{Appendix }

\vspace{1cm}

\textbf{1. The optimal value of $I $ without storage systems}\\\\
Named $I$ the investment in PV, $D$ the environmental damage and $E$ the energy production, we get the following
Hamiltonian of the Problem with State Equations

\vspace{0.5cm}
$$
\begin{array}{ll}

H(t,I, E,D) =
	\\\\

{[(r I_t- c I_t^2)+(\beta E_t -\gamma D_t)]}e^{-\rho t}+ \mu_t[ \eta I_t- \delta_E  E_t ]
+ \lambda_t [ - \omega  \eta  I_t+\delta_D  D_t ]
\end{array}
$$
\vspace{0.5cm}

\vspace{0.5cm}
where $r=r_{pv}/r_f$ , $ c=c_{pv}/c_f$ and $\eta=\eta_{pv}/\eta_f$ represent economic return ,cost and productivity of PV calculated with respect to them of fossil fuels.
\vspace{0.5cm}
$$
\left\{
\begin{array}{ll}
\dot{E_t} = \eta I_t - \delta_E E_t$$
\\\\
\dot{D_t}= -\omega  \eta I_t+ \delta_D D_t
\end{array}
\right.$$

\vspace{0.5cm}
The co-state equations are:

$$
\left\{
\begin{array}{ll}
\dot{\mu_t} =-\frac{\partial H}{\partial E_t}=\delta_E \mu _t-\beta e^{-\rho t}
\\\\
\dot{\lambda_t} =-\frac{\partial H}{\partial D_t}=-\delta_D\lambda _t +\gamma e^{-\rho t}
\end{array}
\right.
$$
with the following transversality conditions

$$
\begin{array}{ll}
\lim_{t \to \infty} \mu_t  = 0
\\\\
\lim_{t \to \infty} \lambda _t = 0
\end{array}
$$
Hence we obtain
$$
\mu _t=\frac{\beta }{\delta_E +\rho }e^{-\rho t}\qquad
\mbox{and}\qquad
\lambda _t=ke^{-\delta_D t}+\frac{\gamma }{\delta_D-\rho  }e^{-\rho t}.
$$
The FOC of the optimization problem is:

$$\frac{\partial H}{\partial I_t} = (r - 2 c I) e^{-\rho t} + \eta(\mu_t  - \omega \lambda_t )=0$$

therefore
$$I_t = \frac{r}{2c} +\frac{\eta }{2c} (\mu_t  - \omega \lambda_t )e^{\rho t}=\frac{r}{2c} +\frac{\eta }{2c}\left [\frac{\beta }{\delta_E +\rho }e^{-\rho t}-\omega \left(ke^{-\delta_D t}+\frac{\gamma }{\delta_D-\rho  }e^{-\rho t}\right )\right ]e^{\rho t}=
$$
$$
=\frac{r}{2c} +\frac{\eta }{2c}\left [\frac{\beta }{\delta_E +\rho }-\omega \left(ke^{(\rho-\delta_D) t}+\frac{\gamma }{\delta_D-\rho  }\right )\right ]
.$$

where, since $I_t$ is bounded, we can set $k=0$ to get the optimal value of $I^*$
$$I^*=\frac{r}{2c} +\frac{\eta }{2c}\left [\frac{\beta }{\delta_E +\rho }+\frac{\omega \gamma }{\rho -\delta_D }\right ]$$

\vspace{0.5cm}
The stationary-state equilibrium values of $E$ and $D$ are:
\vspace{0.5cm}

\begin{equation}
	\begin{array}{ll}
	E^* =\frac{\eta}{\delta_E}\left \{\frac{r}{2c} +\frac{\eta }{2c}\left [\frac{\beta }{\delta_E +\rho }+\omega \left(\frac{\gamma }{\rho -\delta_D  }\right )\right ]\right\}
\\\\
	D^*=\frac{\omega \eta}{\delta_D}\left \{\frac{r}{2c} +\frac{\eta }{2c}\left [\frac{\beta }{\delta_E +\rho }+\omega \left(\frac{\gamma }{\rho -\delta_D  }\right )\right ]\right\}
\end{array}
\end{equation}
\vspace{1cm}

\textbf{2. The optimal value of $I $ with storage systems}

Named $S$ the storage energy, we get the following
Hamiltonian of the Problem with State Equations

$$
\begin{array}{ll}
	H(t,I, E,D, S) =
	\\\\
	={[(r(1-s) I_t- c (1-s)^{2}I _t^2+ (\beta E_t-\gamma D_t+\varsigma S_t-c_S s I_t)]} e^{-\rho t}+\mu [\eta (1-s)I_t - \delta_E E_t]+
	\\\\
	+\lambda [-\omega  \eta (1-s)I_t-q S_t + \delta_D D_t]+\nu [ \eta_S s I_t -\delta_S S_t]
	
\end{array}
$$
s.t.

$$
\left\{
\begin{array}{ll}
	\dot{E_t} = \eta (1-s)I_t - \delta_E E_t
	\\\\
	\dot{D_t}= -\omega  \eta (1-s)I_t - q S_t + \delta_D D_t
	\\\\\
	\dot{S_t}= \eta_S s I_t -\delta_S S_t
\end{array}
\right.
$$
\vspace{0.5cm}

The co-state equations are:
$$
\left\{
\begin{array}{ll}
	\dot{\mu_t} =-\frac{\partial H}{\partial E_t}=\delta_E \mu _t-\beta e^{-\rho t}
	\\\\
	\dot{\lambda_t} =-\frac{\partial H}{\partial D_t}=-\delta_D\lambda _t +\gamma e^{-\rho t}
	\\\\
	\dot{\nu_t}=-\frac{\partial H}{\partial S_t}=\delta_S \nu_t-\varsigma e^{-\rho t}-\lambda_t q 
\end{array}
\right.
$$

with the following transversality conditions

$$\lim_{t \to \infty} \mu_t  = 0$$
$$\lim_{t \to \infty} \lambda _t = 0$$
$$\lim_{t \to \infty} \nu _t = 0$$
Hence solving the costate equations exploiting the transversality conditions, we obtain
$$
\mu _t=\frac{\beta }{\delta_E +\rho }e^{-\rho t},\qquad
\lambda _t=ke^{-\delta_D t}+\frac{\gamma }{\delta_D-\rho  }e^{-\rho t}\qquad
\mbox{and}\qquad
\nu_t=\frac{\varsigma-\frac{\gamma q}{\delta_D -\rho} }{\delta_S +\rho }e^{-\rho t}.
$$

The FOC of the optimization problem is:

$$\frac{\partial H}{\partial I_t} = [r(1-s) - 2 c(1-s)^2 I_t-c_S s] e^{-\rho t} + \eta (1-s)(\mu_t  - \omega \lambda_t )+\eta_S s\nu_t =0$$

Therefore

$$I_1^* = \frac{r(1-s)-c_s S}{2c(1-s)^2} +\frac{\eta (1-s) }{2c(1-s)^2} (\mu_t  - \omega \lambda_t )e^{\rho t}+\frac{ \eta_S s}{2c(1-s)^2}\nu_t.$$

Solving the FOCs for $I$ we finally get the optimal value of $I_1^*$ with storage technologies :

$$
I_1^*=\frac{r(1-s)-c_s s}{2c(1-s)^2}+\frac{1 }{2c(1-s)} \left [\frac{\beta \eta }{\delta_E +\rho }+\frac{\omega  \gamma \eta }{\rho -\delta_D  }+\frac{(\varsigma+\frac{\gamma q}{\rho-\delta_D })\eta_S}{\delta_S +\rho } \right ]
$$

The stationary-state equilibrium values of $E$, $D$ and $S$ are:

$$
\begin{array}{ll}
	\displaystyle {	E^* =\frac{\eta (1-s)}{\delta_E}I_1^*}
	\\\\
	\displaystyle {D^* = \frac{\omega \eta (1-s) }{\delta_D}I_1^* + \frac{q \eta_S s }{\delta_S \delta_D}I_1^*}
	\\\\
	\displaystyle {	S^* = s\frac{\eta_S}{\delta_S}}I_1^*.
\end{array}
$$

\end{document}